\documentclass[11pt]{article} 
\usepackage[english]{babel}
\usepackage[T1]{fontenc}
\usepackage{lmodern} 
\usepackage{graphicx}
\usepackage{wrapfig}
\usepackage{subfig}
\usepackage{caption}
\usepackage{paralist}
\usepackage{enumitem}

\usepackage[usenames]{color} 
\definecolor{hellblau}{rgb}{0.2,0.4,1} 
\definecolor{dunkelblau}{rgb}{0,0,0.8}
\definecolor{dunkelgruen}{rgb}{0,0.5,0}
\usepackage[
	pdftex,
	colorlinks,
	linkcolor=dunkelblau,
	urlcolor=dunkelblau,
	citecolor=dunkelgruen,
	bookmarks=true,
	linktocpage=true,
	pdftitle={},
	pdfauthor={},
	pdfsubject={},
	pdfkeywords={}%
]{hyperref} 
\usepackage{amsmath} 
\usepackage{amsthm} 
\usepackage{amsfonts} 
\theoremstyle{plain} 
	\newtheorem{satz}{Satz}[] 
	\newtheorem{theorem}[satz]{Theorem}
	\newtheorem{lemma}[satz]{Lemma}

\theoremstyle{remark} 
	\newtheorem{conjecture}[satz]{Conjecture}
\theoremstyle{definition} 

	\newtheorem{corollary}[satz]{Corollary}

\title{Counting $K_4$-Subdivisions}
\author{Tillmann Miltzow\footnote{FU Berlin, Germany. Tillmann Miltzow is supported by ERC Starting Grant PARAMTIGHT (No. 280152).} \and Jens M. Schmidt\footnote{TU Ilmenau, Germany.} \and Mingji Xia\footnote{State Key Laboratory of Computer Science, Institute of Software, Chinese Academy of Sciences.}}
\date{}

\begin{document}
\maketitle

\begin{abstract}
A fundamental theorem in graph theory states that any 3-connected graph contains a subdivision of $K_4$. As a generalization, we ask for the minimum number of $K_4$-subdivisions that are contained in every $3$-connected graph on $n$ vertices. We prove that there are $\Omega(n^3)$ such $K_4$-subdivisions and show that the order of this bound is tight for infinitely many graphs. We further investigate a better bound in dependence on $m$ and prove that the computational complexity of the problem of counting the exact number of $K_4$-subdivisions is $\#P$-hard.
\end{abstract}

\section{Introduction}
Subdivisions of the complete graph $K_4$ on four vertices play a prominent role in graph structure theory: They do not only form the inductive anchor for constructive characterizations of 3-connectivity such as Tutte's Wheel Theorem~\cite{Tutte1961} or Barnette and Grünbaum's characterization~\cite{Barnette1969}, they also received considerable attention in the variants of \emph{$K_4^-$-subdivisions} (where $K_4^-$ is a $K_4$ minus one edge) and \emph{totally odd $K_4$-subdivisions} (in which every subdivided edge is of odd length) due to applications for colorings, planarity and parity constrained disjoint paths problems~\cite{Mader1998,Kawarabayashi2010}.

It is folklore that every 3-connected graph contains a $K_4$-subdivision (see~\cite{Barnette1969} for an early reference); with containment we always mean the usual subgraph-relation. In terms of connectivity, this is optimal, since 2-connected graphs do not necessarily contain a $K_4$-subdivision, as cycles or the arbitrarily large graphs $K_{2,n-2}$ show. In terms of numbers, it is optimal, as the minimal $3$-connected graph $K_4$ contains exactly one $K_4$-subdivision.

As a generalization, we ask for the minimum number $\phi(n)$ of pairwise different $K_4$-subdivisions that are contained in every 3-connected graph on $n$ vertices. The dependence on $n$ will allow to prove more than just one such subdivision. We will prove that $\phi(n) \in \Omega(n^3)$ and that this lower bound is tight up to constant factors. We also show that there is a better lower bound when the input graph has many edges, namely $\Omega(m^4/n)$. Finally, we show that the computational problem of counting these $K_4$-subdivisions exactly is $\#P$-hard. Clearly, the \emph{maximal} number of different $K_4$-subdivisions may be exponential in $n$, as the complete graphs show.

\section{Preliminaries}\label{sec:prel}
We will consider only finite and simple graphs. A \emph{subdivision} of a graph $H$ (a \emph{$H$-subdivision}) is a graph obtained from $H$ by replacing every edge with a path of length at least one. A vertex of a $H$-subdivision is called \emph{real} if it has degree at least three in $H$ and \emph{unreal} otherwise.

A \emph{$k$-separator} of a graph $G=(V,E)$ is a set of $k$ vertices whose deletion leaves a disconnected graph. Let $n := |V|$ and $m := |E|$. A graph $G$ is \emph{$k$-connected} if $n > k$ and $G$ contains no $(k-1)$-separator. A path from a vertex $s$ to a vertex $t$ is called an $s$-$t$-path (and contains every vertex at most once). A set of paths is called \emph{independent} if the intersection of the vertices of every two paths is a subset of $\{s,t\}$.

We first give an upper bound for the minimal number of $K_4$-subdivisions in 3-connected graphs with $n$ vertices. Consider a wheel-graph (see Figure~\ref{fig:k4wheel}). Every $K_4$-subdivision of such a graph contains the central vertex as real vertex, as otherwise there are at most two real vertices instead of the desired four. The remaining part of the subdivision is then uniquely defined by choosing $3$ real vertices arbitrarily on the rim of the wheel. This implies the upper bound $\phi(n) \leq \binom{n-1}{3}$ and thus $\phi(n) \in O(n^3)$.

\begin{figure}[!ht]
	\centering
	\includegraphics[width = 0.25\textwidth]{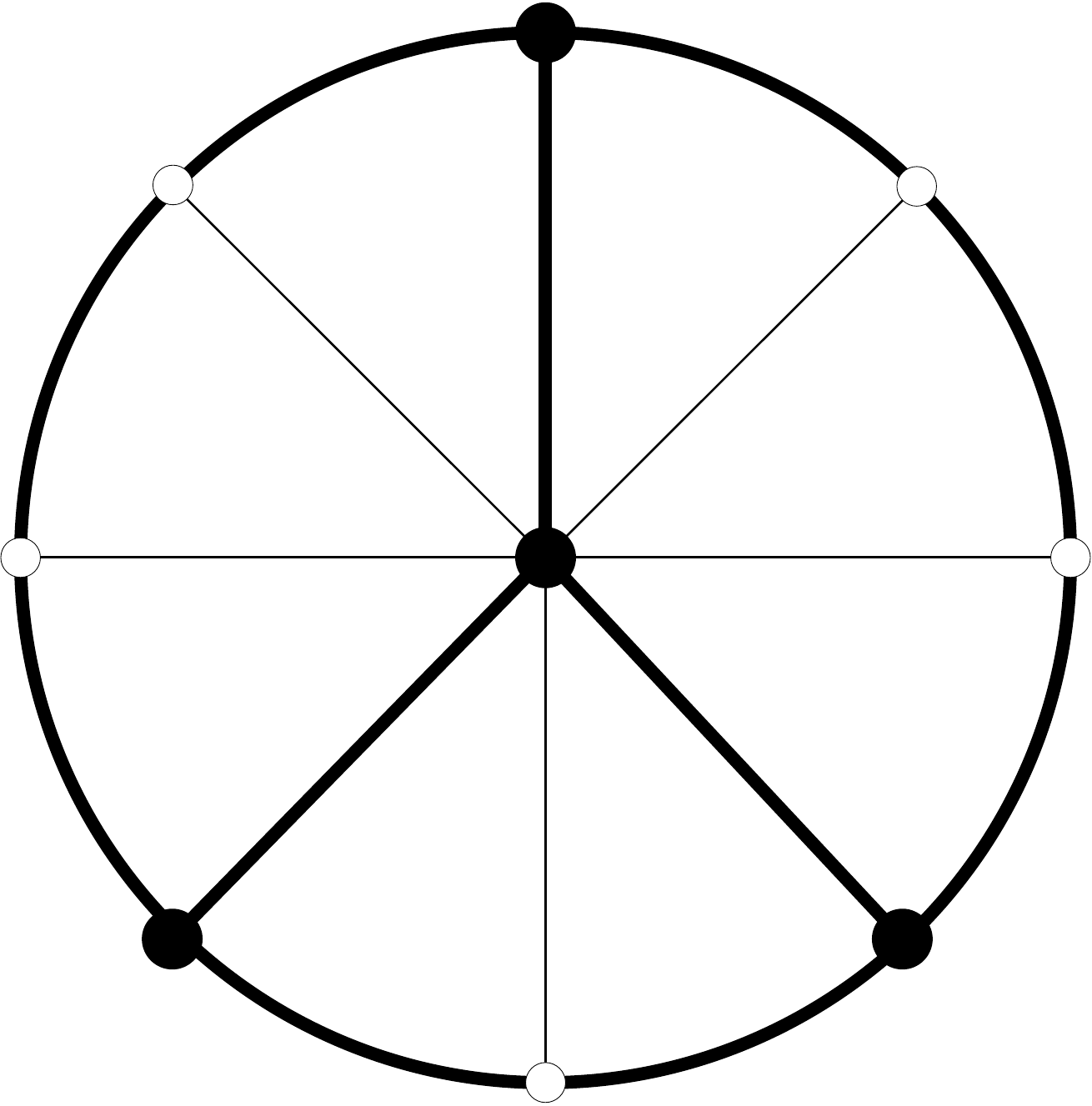}
	\caption{A $K_4$-subdivision of a wheel-graph (fat edges). The black vertices are real vertices.}
	\label{fig:k4wheel}
\end{figure}

For an adequate lower bound of the same order, we will first show some useful facts about the minimum number of cycles in 2-connected graphs.

\section{Cycles in 2-Connected Graphs}
An \emph{open ear decomposition} of a $2$-connected graph $G=(V,E)$ is a sequence $(P_1,P_2,\ldots,P_l)$ of subgraphs of $G$ partitioning $E$ such that $P_1$ is a cycle and every $P_i \neq P_1$ is a path that intersects $P_1 \cup \cdots \cup P_{i-1}$ in exactly its endpoints. Each $P_i$ is called an (open) \emph{ear}~\cite{Lovasz1985,Whitney1932a}.

Open ear decompositions are known to exist for and only for 2-connected graphs. For each $i$, $P_1 \cup \cdots \cup P_i$ is again 2-connected. We will first establish a lower bound on the minimum number of cycles in 2-connected graphs, which is dependent on the number of ears. The first lemma ensures that there are many distinct paths with fixed endvertices.

\begin{lemma}\label{pathlemma}
Let $s$ and $t$ be two vertices in a $2$-connected graph $G$ with $l$ ears. Then $G$ contains $l+1$ distinct $s$-$t$-paths.
\end{lemma}
\begin{proof}
The proof proceeds by induction on the number of ears in an open ear decomposition of $G$. If $l=1$, $G$ is a cycle and the claim follows. If $l>1$, let $G'$ be the $2$-connected graph $P_1 \cup \cdots \cup P_{l-1}$. By induction hypothesis, $G'$ contains $l$ distinct $s$-$t$-paths for any two vertices $s$ and $t$. Let $a$ and $b$ be the two end vertices of $P_l$. We distinguish three cases (see Figure~\ref{fig:Paths}) and prove for each case that $G$ contains an additional $s$-$t$-path.

\begin{enumerate}
	\item $s \in V(G')$ and $t \in V(G')$:\\ It suffices to show that there is an $s$-$t$-path in $G$ that contains $P_l$; this path differs from the other $l$ paths. Consider the graph $H$ that is obtained from $G$ by adding a new vertex $v$ with neighbors $s$ and $t$ and by subdividing an edge of $P_l$ with the vertex $w$. As $H$ is 2-connected, there is a cycle in $H$ containing $v$ and $w$ by Menger's Theorem, which gives the desired $s$-$t$-path containing $P_l$ in $G$.
	\item $s \in V(G')$ and $t \notin V(G')$ (or, by symmetry, vice versa):\\Then $t$ is an inner vertex of $P_l$. By induction, we have $l$ distinct $s$-$a$-paths in $G'$. Extending each of these paths to $t$ along $P_l$ gives $l$ distinct $s$-$t$-paths in $G$. An additional $s$-$t$-path can be obtained by extending an $s$-$b$-path to $t$ along $P_l$.
	\item $s \notin V(G')$ and $t \notin V(G')$:\\ There are $l$ $a$-$b$-paths in $G'$, each of which can be extended to $s$-$t$-paths in $G$. An additional $s$-$t$-path in $G$ is the one in $P_l$.
\end{enumerate}
\vspace{-0.85cm}
\end{proof}

\begin{figure}[!ht]
	\centering
	\includegraphics[width = 1.0\textwidth]{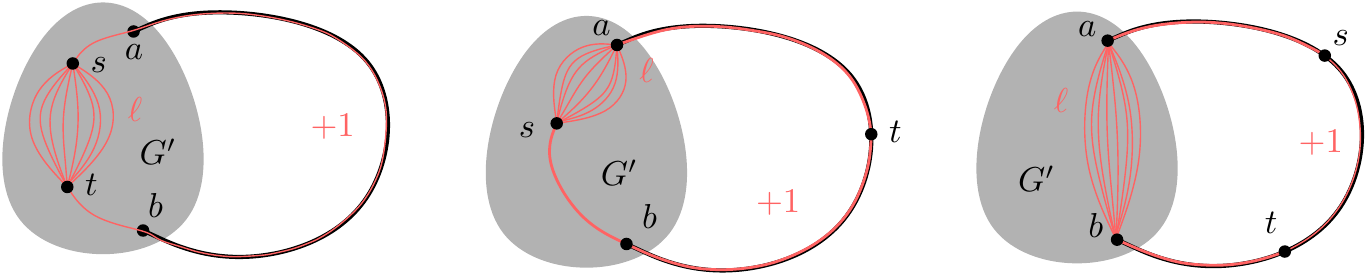}
	\caption{}
	\label{fig:Paths}
\end{figure}

Lemma~\ref{pathlemma} is used to prove the following lower bound on the number of cycles.

\begin{lemma}\label{numbercycles}
Every $2$-connected graph $G$ with $l$ ears contains $\binom{l+1}{2}$ distinct cycles.
\end{lemma}
\begin{proof}
By induction on $l$. If $l=1$, $G$ is a cycle and the claim follows. If $l>1$, let $G'$ be the $2$-connected graph $P_1 \cup \cdots \cup P_{l-1}$. Then $G'$ has $l-1$ ears and contains $\binom{l}{2}$ distinct cycles by induction hypothesis. That are $l$ cycles less than we need to show for $G$. We prove that there are $l$ cycles in $G$, each of which contains $P_l$, which gives the claim. Let $a$ and $b$ be the end vertices of $P_l$. According to Lemma~\ref{pathlemma}, there are $l$ distinct $a$-$b$-paths in $G'$. Augmenting each of these paths with $P_l$ gives the desired $l$ additional cycles.
\end{proof}

Whitney proved that every open ear decomposition has exactly $m-n+1$ ears~\cite{Whitney1932a}. The number $m-n+1$ can be easily obtained by deleting one arbitrary edge from each ear, as the resulting graph will be a tree satisfying $m = n-1$. Applying the number of ears to Lemma~\ref{numbercycles} gives immediately the following corollary.

\begin{corollary}\label{cor}
Every $2$-connected graph contains $\binom{m-n+2}{2}$ distinct cycles.
\end{corollary}

The bound of Corollary~\ref{cor} is tight (for all $n$ and $m = 2n-4$), as the graphs $K_{2,n-2}$ show. If additionally the minimum degree in $G$ is $\delta$, we have $m \geq \delta n/2$ and get the following result.

\begin{corollary}
Every $2$-connected graph with minimum degree $\delta$ contains $\binom{n(\delta / 2 -1)+2}{2} = (\delta-2)^2n^2 /8 + 3(\delta-2)n /4 + 1$ distinct cycles.
\end{corollary}

\section{Counting \texorpdfstring{$K_4$-}{}Subdivisions}
For a vertex $v$ in a $3$-connected graph $G$, let $G_v$ be the graph obtained from $G$ by deleting $v$. Let $d_1,\ldots,d_n$ be the vertex degrees of $G$ (by 3-connectivity, these are at least three) and, for a vertex $v$, let $d_v$ be the degree of $v$ in $G$. Instead of counting $K_4$-subdivisions directly in $G$, we will count cycles in the different graphs $G_v$ and augment these cycles to $K_4$-subdivisions using the following corollary of Menger's theorem.

\begin{lemma}[{Fan Lemma~\cite[Proposition~9.5]{Bondy2008}}]\label{fanlemma}
Let $v$ be a vertex in a $k$-connected graph $G$ and let $C$ be a set of at least $k$ vertices in $G$ with $v \notin C$. Then there are $k$ independent paths $P_1,\ldots,P_k$ from $v$ to distinct vertices $c_1,\ldots,c_k \in C$ such that $V(P_i) \cap C = c_i$ for each $1 \leq i \leq k$.
\end{lemma}

Every cycle $C$ in $G_v$ gives a $K_4$-subdivision of $G$ by applying the Fan Lemma with $v$, $C$ and $k=3$. Every $K_4$-subdivision can occur from at most $4$ graphs $G_v$, as $v$ has to be a real vertex of that $K_4$-subdivision. Thus, each $K_4$-subdivision is counted at most 4 times. We will show that the numbers of cycles for every $G_v$, $v \in V(G)$, sum up to a large value, namely to the value $c \in \Omega(n^3)$. This implies the desired lower bound $\frac{c}{4} \in \Omega(n^3)$ for the number of $K_4$-subdivisions.

It remains to show that $c \in \Omega(n^3)$. Note that each $G_v$ is $2$-connected, as it only differs from $G$ by the deletion of one vertex. Moreover, each $G_v$ has exactly $m-d_v$ edges and $n-1$ vertices. According to Corollary~\ref{cor}, $G_v$ contains at least $\binom{m-d_v-n+3}{2} =: \binom{a-d_v+1}{2}$ cycles, where we define $a := m-n+2$ for brevity. Note that $a-d_v+1 \geq 2$ is positive, since $m-d_v \geq n-1$ because $G_v$ is 2-connected. We calculate the total number of cycles $c$ in all $G_v$ as follows.

\allowdisplaybreaks[1]
\begin{align*}
c &\geq \sum^n_{i=1}\binom{a-d_i+1}{2}\\
&= \sum^n_{i=1} \frac{a^2-2ad_i+d^2_i+a-d_i}{2}\\
&\geq \frac{1}{2}na^2 + \frac{1}{2}na - m - a\sum^n_{i=1}d_i +\frac{1}{2}\sum^n_{i=1}d^2_i\tag{as $\sum^n_{i=1}{d_i}=2m$}\\
&\geq \frac{1}{2}na^2 + \frac{1}{2}na - m - 2ma + 2m^2/n\tag{$\star$, Cauchy-Schwarz}\\
&= \frac{1}{2n}(m(n-2) -n^2+2n)^2 + \frac{1}{2}na - m\\
&\geq \frac{1}{2n} \left( \frac{3}{2}n^2 -3n -n^2+2n \right)^2 + \frac{1}{2}na - m\tag{as $m \geq \frac{3}{2}n$}\\
&= \frac{1}{8}n^3 - \frac{1}{2}n^2 + \frac{3}{2}n + \frac{1}{2}m(n-2) - \frac{1}{2}n^2\\
&\geq \frac{1}{8}n^3 - \frac{1}{2}n^2 + \frac{3}{2}n + \frac{3}{4}n^2 - \frac{3}{2}n - \frac{1}{2}n^2 \tag{as $m \geq \frac{3}{2}n$}\\
&= \frac{1}{8}n^3 - \frac{1}{4}n^2\\
\end{align*}

For $\star$, we used that $\sum^n_{i=1} d^2_i \geq \frac{\left(\sum^n_{i=1} d_i\right)^2}{n} = \frac{4m^2}{n}$, as $\sum^n_{i=1} d_i \leq \sqrt{n}\sqrt{\sum^n_{i=1} d^2_i}$ follows directly from applying the Cauchy-Schwarz inequality to the all 1- and the degree-vector. Using the upper bound of Section~\ref{sec:prel}, we obtain the following theorem.

\begin{theorem}\label{thm:boundN}
For every $n$, $\frac{1}{32}n^3 - \frac{1}{16}n^2 \leq \phi(n) \leq \frac{1}{6}n^3-n^2+\frac{11}{6}n+1 = \binom{n-1}{3}$. Thus, $\phi(n) \in \Theta(n^3)$.
\end{theorem}

We conjecture that the upper bound coming from the wheel graphs is actually the right bound.

\begin{conjecture}
For every $n$, $\phi(n) = \binom{n-1}{3}$.
\end{conjecture}

\section{A Better Bound for Large m}
When $m$ is large, we can obtain better lower bounds in dependence on $m$. Let $\phi(n,m)$ be the minimum number of pairwise different $K_4$-subdivisions that are contained in every 3-connected graph having $n$ vertices and $m$ edges. We use the same idea as above, but additionally construct many $K_4$-subdivisions from one cycle $C$ in $G_v$ whenever $d_v$ is large.

\begin{figure}[!ht]
	\centering
	\includegraphics[width = 0.7\textwidth]{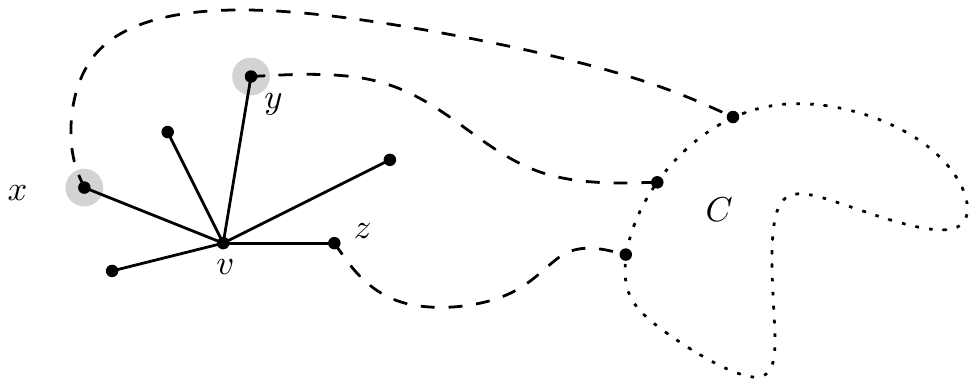}
	\caption{}
	\label{fig:quadraticCycles}
\end{figure}

Let $D := \{v,x,y\}$, where $x$ and $y$ are arbitrary distinct neighbors of $v$ in $G$ ($D$ may intersect $C$). By Menger's theorem, there are three independent $C$-$D$-paths in $G$, possibly of length 0, which we extend to three independent $C$-$v$-paths by adding the edges $xv$ and $yv$. This forces two of the three independent paths to go through $x$ and $y$ (see Figure~\ref{fig:quadraticCycles}). The third neighbor $z$ of $v$ in the independent paths cannot be forced this way. However, we may obtain the same $K_4$-subdivision containing $v,x,y,z$ and $C$ for $G_v$ only when $D$ is either $\{v,x,y\}$, $\{v,x,z\}$ or $\{v,y,z\}$. Thus, such a $K_4$-subdivision is counted at most three times. Since we can choose any two neighbors $x$ and $y$ of $v$, this gives at least $\lceil \frac{1}{3}\binom{d_v}{2} \rceil$ $K_4$-subdivisions for every cycle $C$ in $G_v$. Thus, we obtain

\allowdisplaybreaks[1]
\begin{align*}
\phi(n,m) &\geq \frac{1}{4} \sum^n_{i=1} \left\lceil \frac{1}{3}\binom{d_i}{2} \right\rceil \binom{a-d_i+1}{2} \tag{$\star$}\\
\end{align*}

We will use this inequality for the following theorem.

\begin{theorem}
$\phi(n,m) \in \Omega(m^4/n)$.
\end{theorem}
\begin{proof}
If $m < 3n$, $m = \Theta(n)$ and the bound follows directly from $\phi(n) \in \Theta(n^3)$ of Theorem~\ref{thm:boundN}. Thus, let $m > 3n$; if we can prove the same asymptotic bound for this case, taking the lower constant factor of the two cases yields the claim. We have $a-d_i+1 \geq m-2n+3 > \frac{m}{3}$ in inequality $(\star)$. Thus,

\allowdisplaybreaks[1]
\begin{align*}
\phi(n,m) &\geq \frac{1}{12} \sum_{i=1}^{n} \binom{d_i}{2} \binom{\frac{m}{3}}{2} = \frac{1}{24} \binom{\frac{m}{3}}{2} \left(\sum^n_{i=1}d^2_i - \sum^n_{i=1}d_i \right)\\
&\geq \frac{1}{24} \binom{\frac{m}{3}}{2} \left(4\frac{m^2}{n} - 2m \right) \tag{Cauchy-Schwarz, $\sum^n_{i=1}{d_i}=2m$}\\
&= \frac{m}{6} \binom{\frac{m}{3}}{2} \left(\frac{m}{n} - \frac{1}{2} \right)\\
&\in \Omega(m^4/n) 
\end{align*}
\end{proof}

\section{\#P-Hardness}
Instead of only giving a lower bound for the number of $K_4$-subdivisions, one may try to compute their exact number. Let \#\texttt{SUBDIVISIONS} be the problem of counting the exact number $\#K_4(G)$ of $K_4$-subdivisions in general graphs. We reduce the following \#P-hard problem \#\texttt{S-T-PATHS}~\cite{Valiant1979} to the problem \#\texttt{FIXED-SUBDIVISIONS} and then to \#\texttt{SUBDIVISIONS}, which proves that it is \#P-hard.

\begin{enumerate}[itemsep=0pt,leftmargin=*,widest=Problem:]
	\item[Problem:] \#\texttt{S-T-PATHS} (this problem is \#P-hard~\cite{Valiant1979}) 
	\item[Input:] $G$; $s,t \in V$
	\item[Output:] Number of different $s$-$t$-paths in $G$.
\end{enumerate}

\begin{enumerate}[itemsep=0pt,leftmargin=*,widest=Problem:]
	\item[Problem:] \#\texttt{FIXED-SUBDIVISIONS} 
	\item[Input:] $G$; $a,b,c,d \in V$
	\item[Output:] Number of different $K_4$-subdivisions in $G$ having $a,b,c,d$ as real vertices.
\end{enumerate}

\begin{theorem}\label{thm:hard}
Counting $K_4$-subdivisions in general graphs is \#P-hard.
\end{theorem}
\begin{proof}
We first reduce \#\texttt{S-T-PATHS} to \#\texttt{FIXED-SUBDIVISIONS}. Given an input $(G,s,t)$ of the first problem, construct the input $(G',a,b,c,d)$ for the second problem such that $G'$ is obtained from a $K_4$ with vertices $\{a,b,c,d\}$ by replacing the edge $ab$ with the graph $as \cup G \cup tb$ (see Figure~\ref{fig:reduction}). Thus, $G$ contains an $s$-$t$-path if and only if $G'$ contains an $a$-$b$-path not intersecting $\{c,d\}$. It follows that the number of $s$-$t$-paths in $G$ is exactly the number of $K_4$-subdivisions having real vertices $\{a,b,c,d\}$ in $G'$.

\begin{figure}[!ht]
	\centering
	\includegraphics[width = 0.9\textwidth]{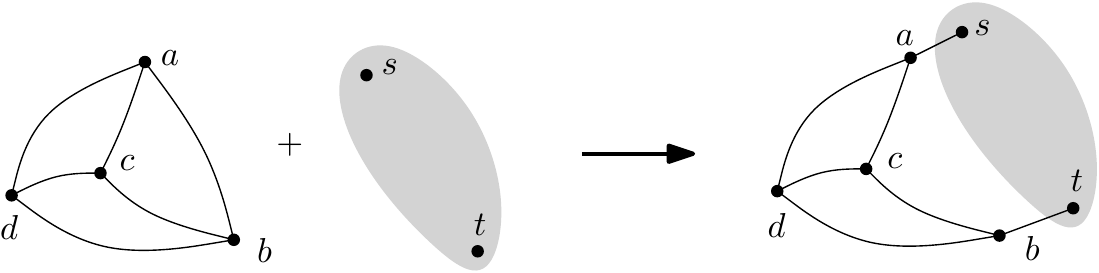}
     \caption{}
	\label{fig:reduction}
\end{figure}

We now reduce \#\texttt{FIXED-SUBDIVISIONS} to \#\texttt{SUBDIVISIONS}. Suppose $(G',a,b,c,d)$ is an instance of the first problem; we construct the instance $G''$ of the second problem by replacing certain edges of $G'$ with the gadget shown in Figure~\ref{fig:weightedge}. The number of cycles in this gadget is fixed to $s := n^2$, so that $2^s$ exceeds the maximal number $\#K_4(G)$ of $K_4$-subdivisions in $G'$ (which is at most $2^{\binom{n}{2}}$). Clearly, the sizes of $G'$ and $G''$ are polynomial in the size of $G$.

\begin{figure}[!ht]
	\centering
	\includegraphics[width = 0.6\textwidth]{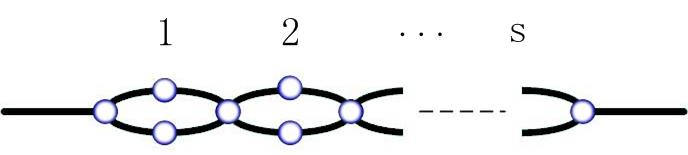}
     \caption{}
	\label{fig:weightedge}
\end{figure}

To construct $G''$ from $G'$, we replace every edge having exactly one endvertex in $\{a,b,c,d\}$ by one gadget and every edge having both endvertices in $\{a,b,c,d\}$ by two gadgets joined in series. These gadgets allow thus $2^s$ and $2^{2s}$ different paths between their endvertices, respectively. For convenience, we may see $G''$ as weighted graph $G'$ for which each edge is weighted with either $1$, $2^s$ or $2^{2s}$.

Clearly, no inner vertex of a gadget may be a real vertex of a $K_4$-subdivision. Thus, there is the identity mapping between the $K_4$-subdivisions in $G'$ and the (weighted) ones in $G''$.
Suppose there are $N_{x,y}$ $K_4$-subdivisions in $G'$ with exactly $x$ real vertices in $\{a,b,c,d\}$ and exactly $y$ unreal vertices in $\{a,b,c,d\}$. Each of them corresponds to a weighted $K_4$-subdivision of $G''$, for which the product of edge weights is exactly $2^{s(3x+2y)}$.
Hence, it corresponds to exactly $2^{s(3x+2y)}$ different $K_4$-subdivisions in $G''$.
In fact, $\#K_4(G'')=\sum_{x+y\leq 4}2^{s(3x+2y)}N_{x,y}$.
As $2^s$ exceeds $N_{x,y}$, $\left\lfloor \frac{\#K_4(G'')}{2^{12s}} \right\rfloor=N_{4,0}$ is the answer for instance $(G',a,b,c,d)$.
\end{proof}

While giving \#P-hardness, the above reductions only argue about general graphs. Using the result above, we show the stronger statement that counting $K_4$-subdivisions in $k$-connected graphs is still \#P-hard for every fixed $k$.

\begin{theorem}
For any fixed $k$, counting $K_4$-subdivisions in $k$-connected graphs is \#P-hard.
\end{theorem}
\begin{proof}
We can assume $k > 1$, as the arguments in the proof of Theorem~\ref{thm:hard} hold also for connected graphs. Let $G$ be an instance of \#\texttt{SUBDIVISIONS} and let $\{v_1,\ldots,v_n\}$ its vertex set. For a reduction to the problem in question, we construct instances $G_s$ from $G$ by adding $s > n$ new vertices $\{x_1,\ldots,x_s\}$ and all edges $x_iv_j$ for $1 \leq i \leq s$ and $1 \leq j \leq n$. Clearly, $G_s$ is $k$-connected and $n \geq 3$, as $n > k > 1$.

Consider a $K_4$-subdivision of $G_s$. It contains at most 4 real $x_i$-vertices and at most $n$ $v_j$-vertices. Thus, it contains at most $3n/2$ unreal $x_i$-vertices (this is not the best possible bound). In total, it contains at most $3n$ of the $x_i$-vertices, since $3n/2 \geq 4$.

In every $K_4$-subdivision of $G_s$, we delete all $x_i$-vertices and call the remaining graph a \emph{partial} $K_4$-subdivision. Let $N_t$, $0 \leq t \leq 3n$, be the number of different partial $K_4$-subdivisions of $G_s$ that were generated by deleting exactly $t$ vertices. For the desired $k$ and every integer $r \geq 0$, the number $N_t$ of $G_s$ with $s := (k+3n)+r$ is the same, by interchangeability of the $x_i$-vertices.

For the same reason, each partial $K_4$-subdivision that was counted for $N_t$ can be extended to a $K_4$-subdivision of $G_s$ in a number $P_{s,t}$ of different ways that is only dependent on $s$ and $t$: Namely, $P_{s,t} = s!/(s-t)! = s(s-1)\cdots(s-t+1)$, which is the number of ways we can choose $t$ ordered non-repetitive elements from $\{x_1,\ldots,x_s\}$. Hence,
$$\#K_4(G_s)= \sum_{t=0}^{3n} P_{s,t} N_t=N_0+\sum_{t=1}^{3n} s(s-1)\cdots(s-t+1) N_t\linebreak=N_0+\sum_{t=1}^{3n} s^t N'_t,$$
where $N'_t$ is some fixed linear combination of $N_t,N_{t+1},\ldots,N_{3n}$ that is not dependent on $s$.

We construct the graph $G_s$ for each $s \in \{k+3n,k+3n+1,...,k+6n+1\}$ and obtain $3n+1$ linear equations in $N_0,N'_1,\ldots,N'_{3n}$ (we set $N'_0 := N_0$), whose coefficient matrix $M$ is Vandermonde in $s$ and nonsingular, as all the values of $s$ are pairwise distinct. We thus have the equation $\overrightarrow{\#K_4(G_s)} = M \overrightarrow{N'_t}$ for the corresponding vectors. As $M$ is nonsingular, we can invert it and obtain $\overrightarrow{N'_t} = M^{-1}\overrightarrow{\#K_4(G_s)}$. As we know the elements of $\#K_4(G_s)$ for all $s$, we get $\overrightarrow{N'_t}$ and therefore in particular $N_0$, which is equal to $\#K_4(G)$.
\end{proof}

\bibliography{Jens}
\bibliographystyle{abbrv}
\end{document}